\documentclass[aps,pra,preprint,longbibliography,nofootinbib]{revtex4-2}

\makeatletter
\renewcommand{\p@subsection}{}
\renewcommand{\p@subsubsection}{}
\makeatother

\usepackage[utf8]{inputenc}
\usepackage[T1]{fontenc}
\usepackage{amsmath, amsfonts, amssymb, amsthm,mathtools}
\usepackage{booktabs}
\usepackage{graphicx}
\usepackage{hyperref}
\usepackage{mathtools}
\usepackage{physics}
\usepackage{enumitem}
\usepackage{listings}
\usepackage{xcolor}
\hypersetup{colorlinks=true, urlcolor=blue}
\usepackage[english]{babel}
\usepackage{hyperref}

\theoremstyle{plain}
\newtheorem{theorem}{Theorem}
\newtheorem*{theorem*}{Theorem}

\newtheorem{corollary}{Corollary}
\theoremstyle{definition}
\newtheorem{definition}{Definition}
\theoremstyle{remark}
\newtheorem{remark}{Remark}

\newcommand{\G}[0]{\mathcal{G}}

\newcommand{\I}[0]{\mathcal{I}}
\newcommand{\J}[0]{\mathcal{J}}
\newcommand{\K}[0]{\mathcal{K}}
\renewcommand{\L}[0]{\mathcal{L}}

\newcommand{\Eb}[0]{\mathbb{E}}

\newcommand{\Ib}[0]{\mathbb{I}}

\newcommand{\Rb}[0]{\mathbb{R}}

\renewcommand{\dd}{\mathrm{d}}

\begin{document}

\title{Formalization of the generalized Pareto principle and structural typicality of the 20/80-rule}
\author{Antti Hippel\"ainen}
\email{antti.hippelainen@helsinki.fi}
\affiliation{Department of Physics and Helsinki Institute of Physics,
P.O. Box 64, FI-00014, University of Helsinki, Finland}

\begin{abstract}
We formalize a generalized form of the Pareto principle --- ``fraction $p$ of inputs yields fraction $1-p$ of outputs'' --- as a property of non-negative gain densities $\ell \in L^1([0,1])$, working with the decreasing rearrangement to obtain a unique characterization. For probability distributions, the resulting $p$ coincides with $1 - k_F$, where $k_F$ is the Kolkata index of the corresponding Lorenz curve. Within this framework we analyze both constructed gain densities and commonly encountered distribution families. We derive closed-form expressions for $p$ for truncated power-law, exponential, and normal distribution families. Combining these with estimates of the truncation parameter as a function of sample size $N$, we predict that datasets of size $N \in [10^2, 10^5]$ from exponential and normal families concentrate $p$ near $[0.15, 0.26]$ and $[0.20, 0.29]$ --- values close to the canonical 0.2/0.8-rule, and strictly below the saturation $k \approx 0.865$ conjectured earlier by Ghosh and Chakrabarti. We discuss the implications of the structural ubiquity of Pareto-type imbalances for their use as
prescriptive targets.
\end{abstract}

\maketitle
\clearpage

\section{Introduction}
\label{sec:intro}

The so-called Pareto principle or ``20/80--rule'' is among the most widely quoted heuristics in economics, management, and cognitive science. It states that 20\% of causes result in 80\% of effects. Originally formulated by Vilfredo Pareto \cite{Pareto:1897}, it was an empirical observation about the distribution of wealth, later generalized to domains as diverse as innovation output, learning efficiency, online participation, frequency of word use, scientific productivity, and urban population size \cite{Nielsen:2006, Zipf:1950, Lotka:1926, Merton:1968, Newman:2005}.

The principle is often used in two roles. Descriptively, it reflects the concentration of effort and reward. Prescriptively, it underlies rules of thumb such as ``focus on the vital few''. For example, cognitive agents reduce complexity by focusing on the few high-impact variables that account for most effects by what Herbert Simon called bounded rationality \cite{Simon:1955}. Such heuristics can be useful if Pareto-like behavior is present.

Given such prominence, it is natural to ask how special and likely the Pareto principle really is. Is a 20/80-balance a contingent empirical regularity or a purely structural feature? If such asymmetries are unavoidable, which specific balances can be expected, and what is the meaning of the principle?

In this work, we address these questions by treating generalized Pareto-type statements as properties of abstract gain distributions. We formalize bounded cumulative processes and illustrate the framework with explicit gain densities that capture both specially constructed cases and common distribution families. The constructed cases clarify which technical restrictions are needed and allow us to reason about the minimum level of inequality required to satisfy a given generalized principle. For distribution families, we express the generalized principle in terms of a small number of parameters and identify which forms of the principle can hold; for some common distributions, we cannot realize arbitrary generalized Pareto principles, highlighting that the set of attainable imbalances depends subtly on the underlying distributional form.

Combining these results with an estimate for typical dataset sizes and parameter estimates, we show that generalized principles numerically close to 20/80 arise quite naturally for phenomena following exponential- and normal-distributions, while heavy-tailed distributions like power-laws tend to produce smaller values of $p$. Taken together, these results support an interpretation of Pareto-type imbalances as structural features of gain distributions rather than as rare empirical curiosities. At the same time, they suggest caution in using the canonical 20/80 rule as a universal yardstick or as a sole justification of inequality. 

The rest of the article is organized as follows. In Section \ref{sec:formalization}, we set the mathematical framework. Section \ref{sec:examplesexistence} applies this framework to concrete density profiles and studies the existence of the generalized principle. Section \ref{sec:discussion} studies commonly encountered generalized principles and discusses the social connections. Clarifying discussions and longer proofs are relegated to appendices \ref{app:limitationsoff} and \ref{app:rearrangementsproof}. We conclude in Section \ref{sec:conclusion}.

\section{Formalization}
\label{sec:formalization}

We model bounded cumulative processes with a non-negative gain density
\begin{equation}
\label{eq:gaindensity}
    \ell : I_t \to [0,\infty), \qquad  \int_{I_t} \dd t \ \ell(t) = 1 \ ,
\end{equation}
where $I_t = [t_\text{min},t_\text{max}]$ is a compact (closed and bounded) interval. Since total gains are assumed to be finite, we require $\ell \in L^1\qty(I_t)$, and that $\ell$ is normalized to unity. In fact, if either the domain of input or the total output were not finite, the (generalized) Pareto principle would be meaningless.

Since we are interested in a question of ratios of gains and $I_t$ is compact, we assume without loss of generality that $I_t = [0,1]$. If the underlying distribution is ill-defined at some value of the original variable (\textit{e.g.} a power-law at $t = 0$), we use $t$ only as a parameter in an affine map; this will be made explicit in Section \ref{sec:examplesexistence}.

The cumulative gain function corresponding to $\ell$ is
\begin{equation}
    \L : [0,1] \to [0,1], \qquad  \L(t) = \int_{0}^t \dd t' \ \ell(t') \ .
\end{equation}
By definition $\L(0) = 0$, $\L(1) = 1$, and $\L$ is (absolutely) continuous:
\begin{equation}
    \left|\L(t) - \L(s) \right| = \left| \int_s^t \dd t' \ \ell(t') \right| \leq \int_s^t \dd t' \ \left| \ell(t') \right| \longrightarrow_{s \to t} 0 \ .
\end{equation}
These definitions mirror those of probability density functions and cumulative distribution functions, but other forms of $\ell$ that do not permit a probabilistic interpretation could in principle exist, and are discussed further in Appendix \ref{app:limitationsoff}. Basic facts of probability and measure theory will still be useful in what follows.

\subsection{Preliminary formulation of the generalized principle}
Tentatively, we let the generalized Pareto principle mean, that a fraction $p$ of causes/inputs/efforts results in a fraction $1-p$ of effects/outputs/gains. Mathematically, this translates to the \emph{preliminary} formulation
\begin{equation}
\label{eq:generalprincipleunion}
\int_{\cup_{j \in \J} I_j} \dd t \ \ell(t) = 1-p = \sum_{j \in \J} \qty(\L(b_j) - \L(a_j))\ ,
\end{equation}
where $\{I_j\}_{j \in \J}$ is a countable collection of disjoint intervals such that for all $j, \ I_j = [a_j, b_j]\subset [0,1]$, and the total length of the union is fraction $p$: $\mu\qty(\cup_{j \in \J} I_j)= p$. By symmetry, we may restrict without loss of generality to $p \in (0, 1/2]$. Whenever referring to a generalized principle with given $p$, this $p$ refers to the first value in $p/(1-p)$. We refer to ``fraction $p$ of inputs yields fraction $1-p$ of outputs'' as the $p/(1-p)$-principle, the generalized (Pareto) principle, or later just ``the principle''. 

Sometimes, the generalized principle (which in our definition would rather be an asymmetric generalized principle) is rather used to mean ``fraction $p$ of inputs yields fraction $1 - k p$ of outputs'', where $k > 0$ is a given ``Pareto parameter''. Later on we will focus only on the symmetric case with $k = 1$, but it is equally easy to keep $k$ free in the following existence property:
\begin{remark}
\label{rmrk:remark1}
    For any continuous gain function $\L : [0,1] \to [0,1] \text{ with } \L(0) = 0 \text{ and } \L(1) = 1$, and for all $k \in \Rb^+$, define $\G(t) = \L(t) - 1 + kt$. Since $\G(0) = -1$ and $\G(1) = k > 0$, the intermediate value theorem guarantees the existence of a $p^* \in (0,1)$ with $\L(p^*) = 1 - k p^*$. Hence, some form of the asymmetric generalized $p/(1 - kp)$-principle is satisfied. In particular, with $k = 1$, there exists $p^*$ such that $\L(p^*) = 1- p^*$ for some $p^* \in (0,1)$. The substantive question becomes which $p^*$ arises, not whether one exists.
\end{remark}

The existence guaranteed by the above remark is not new: an equivalent fixed-point condition has been studied in the context of economic inequality under the name Kolkata index with $k = 1$ \cite{Ghosh:2014,Banerjee:2019,Banerjee:2020,Ghosh:2021}. The index has been defined as the unique solution to the equation $L(k_F) + k_F = 1$, where $L$ is the Lorenz curve. Effectively, our formulation agrees with the definition of the Kolkata index on probability densities, and extends the definition to $L^1$-integrable densities with less stringent assumptions (see Appendix~\ref{app:limitationsoff} and Section~\ref{sec:conclusion} for discussions of generalizations). Acknowledging this, we shall retain the terminology of a generalized Pareto principle throughout, both to emphasize the input/output framing that doesn't necessarily have anything to do with inequality, and to cover the possibilities of generalization. Readers familiar with the Kolkata index may read $p^*$ as $1 - k_F$ for most purposes.

Also, in \cite{Ghosh:2021} the existence of a general social constant was conjectured based on the study of the Kolkata index across multiple datasets, and the one-parameter family of Lorenz-curves with $L(x) = x^n$ \cite{Ghosh:2021}. The present work extends the study of the existence of such a social constant by studying a variety of distributional families with more parameters. We also derive rigorous bounds on the possible generalized Pareto principles satisfied. As we will see, the results from these derivations are often below the previously-conjecture bound of $0.86$.

Since satisfying the (asymmetric) generalized principle only requires $\L$ to be continuous, we have allowed $\ell \in L^1([0,1])$, and so under very mild assumptions a bounded cumulative process satisfies the generalized Pareto principle for some $p$; by \ref{rmrk:remark1}, the question becomes, can we satisfy a \textit{specific} form of the generalized principle with a \textit{specific} choice of $k$? Since any $k$ is as good as any other, we fix $k = 1$ moving forward.

The preliminary naive definition \eqref{eq:generalprincipleunion} allows us to choose an arbitrary countable disjoint union of intervals of total length $p$. However, this definition turns out to be too permissive and trivializes the existence of the generalized principle. Before the proof, intuitively, definition \eqref{eq:generalprincipleunion} allows us to choose any regions of values of $\ell$ to contribute to the generalized principle. Even if choosing specific parts of a distribution may seem arbitrary, this would allow us to pick the most unequal combinations of inputs/outputs from the original distribution, which is at the heart of studying the imbalance predicted by the (generalized) Pareto principle. We formalize this intuition with the construction of a decreasing rearrangement for the gain density, resulting in the following theorem, the proof of which can be found from Appendix~\ref{app:rearrangementsproof}:
\begin{theorem}
\label{thm:theorem2}
        For a gain density $\ell \in L^1([0,1])$, allowing the generalized Pareto principle to be satisfied over any countable disjoint union of intervals $\bigcup_{j \in \J} I_j, \ \forall \ j \ I_j \subset [0,1]$, is equivalent to fixing an interval of integration of length $\sum_j \mu(I_j)$ and allowing arbitrary rearrangements of $\ell$ into some $\ell^*$, which is equimeasurable with $\ell$. This also results in the rearranged gain function $\L^*$. \\

    \noindent With such rearrangements allowed, the minimal $p$ for satisfying the generalized principle is obtained by integrating $\ell^*$ on an interval $[0,p^*]$ such that $\L^*(p^*) = 1- p^*$. By symmetry, the maximal value of $p$ satisfying the generalized principle is $1 - p^*$. \\

    \noindent Then, every generalized Pareto principle between $p^*$ and $(1-p^*)$ can be satisfied by some rearrangement.
\end{theorem}
From the above immediately follows:
\begin{corollary}
    Under the preliminary definition \eqref{eq:generalprincipleunion}, a given distribution satisfies the exact $0.2/0.8$-principle if and only if its decreasing rearrangement satisfies some $p/(1-p)$-principle with $p \leq 0.2$.
\end{corollary}

Therefore, using \eqref{eq:generalprincipleunion} as a definition for the generalized Pareto principle allows a family of generalized principles to be satisfied with rearrangements of the same underlying distribution. This is conceptually too weak, and in order to study a unique generalized principle with non-trivial behavior, we must impose further limitations.

\subsection{Final formulation of the generalized principle}

Considering the conceptual meaning of the original Pareto principle as well as the more general form, we impose two natural requirements to obtain a unique and non-trivial characterization:
\begin{enumerate}
\itemsep0.1em
    \item \textbf{Maximal discrepancies} \\
    Pareto-type statements are typically used to quantify how far a system is from a 0.5/0.5-balance. It is therefore natural to focus on the maximal discrepancy between fractions of input and output.
    \item \textbf{Domain order-independent formulation} \\
    There often exists an ordering of the domain (time, size, wealth, etc.), and hence, the distribution. To obtain an ordering independent of the original dataset, we ``sort'' the gain density from largest values to smallest by a decreasing rearrangement.
\end{enumerate}

Implementing these leads to the following working definition:
\begin{definition}
    Let $\ell :[0,1] \to [0,\infty)$ be $L^1([0,1])$ and normalized to unity, and let $\ell^*$ denote its decreasing rearrangement. Define the cumulative gain function of such distribution as 
    \begin{equation}
        \L^{(*)}(t) = \int_0^t \dd s \ \ell^{(*)}(s) \ .
    \end{equation}
    Then, $\ell$ satisfies a generalized Pareto principle with $p \in (0, 1/2]$ if 
    \begin{equation}
        \L^*(p) = 1- p \ .
    \end{equation}
\end{definition}

We restrict to $p \in (0,1/2]$ for two reasons: First, by symmetry, a gain density satisfying the $p/(1-p)$-principle also satisfies the $(1-p)/p$-principle, so $p > 1/2$ is redundant. Second, the endpoints $p = 0,1$ correspond to degenerate cases (discrete measures) which lie outside $L^1$, and for all practical purposes, it is sufficient to get arbitrarily close to these limits. Working on a half-open interval also allows for easy inequality proofs in what follows.

In data, either the underlying random variable is discrete, or the data must be binned to find a sensible distribution. An additional factor of variation in the $L^1$-approximation of data results from choosing the (width of) bins. However, as long as an $L^1$-approximation is used and a decreasing rearrangement performed, the existence of a unique generalized Pareto principle is guaranteed.

Partially due to the above freedom, the principle is often not considered in the strictest sense with data, but rather an agreement up to a non-well-defined tolerance is accepted -- for example, a ``0.20/0.78''-rule might be considered ``good enough'' to count as an occurrence of the canonical Pareto principle. Since we are working with $L^1$-(approximated) distributions, we will study only the exact generalized Pareto principles and discuss the special role of the canonical 0.2/0.8-principle further in Section \ref{sec:discussion}.

Finally, we have decided to frame the principle without explicit reference to standard probability theory, since we do not want to impose upon ourselves the related mathematical restrictions\footnote{In fact, historically, Pareto himself was aware of the possibility of a probabilistic interpretation, but decided to frame the principle without the theory of probabilities \cite{Tusset:2024}.}.
This is simply because there could exist sensible interpretations of other gain densities $\ell$ than those allowed by probability theory; see Appendix \ref{app:limitationsoff} for further discussion. The interested reader may find some of the measure-theoretic considerations laid out explicitly in \cite{Hardy:2010}.

\section{Example distributions and existence of generalized principles}
\label{sec:examplesexistence}

We now examine gain density examples to illustrate how the generalized principle emerges in diverse functional forms. These cases demonstrate, in addition to the framework, the limitations imposed on $\ell$. For each distribution family we also derive which range of generalized principles can be satisfied by suitably tuning the underlying parameters. 

\subsection{Constructed examples}
\label{sec:constructedcases}

\subsubsection{Step function}

Consider a step function profile: a period of faster gains followed by a period of slower gains. For the canonical 0.2/0.8-principle, 
\begin{equation}
    \ell(t)=
\begin{cases}
4, & 0\le t < 0.2 \\
0.25, & 0.2\le t\le 1 
\end{cases}    \ .
\end{equation}
Such step function densities with $p = 0.1, 0.2$ and 0.4, and the associated cumulative distributions are shown in Fig. \ref{fig:step}. 
\begin{figure}[ht]
    \centering    
    \includegraphics[width=\linewidth]{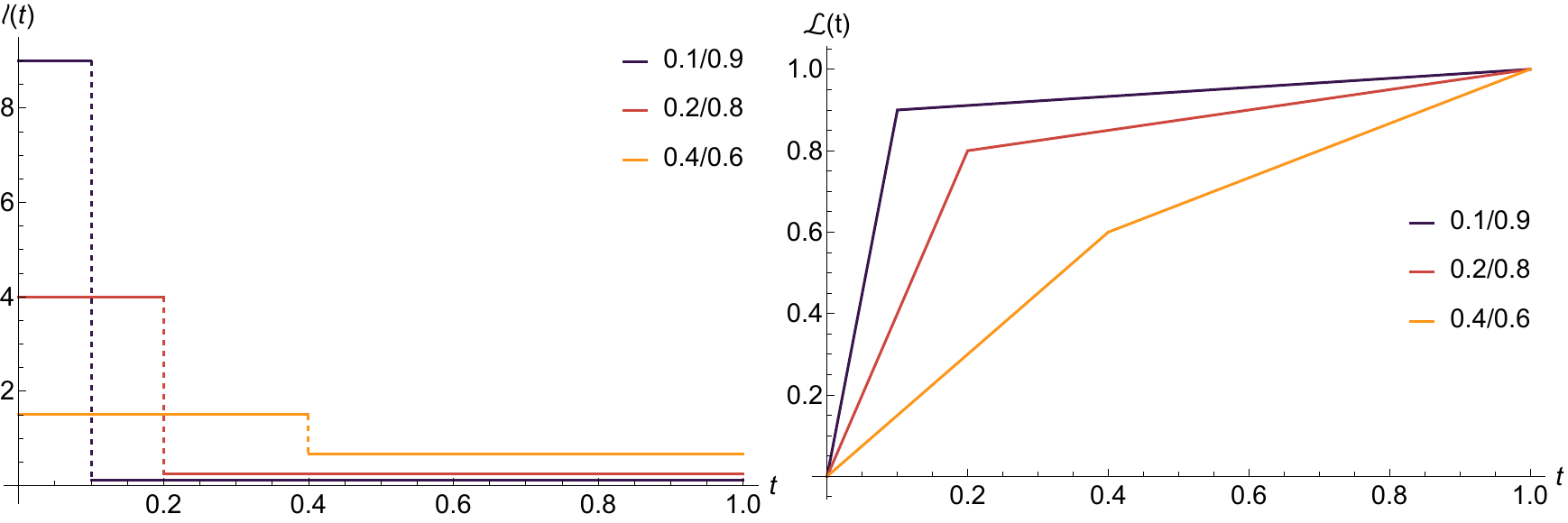}
    \caption{Step function densities and their related cumulative distributions with $p = 0.1, 0.2$ and $0.4$.}
    \label{fig:step}
\end{figure}

Even though otherwise trivial, such densities have an interesting feature. Define how egalitarian a distribution is by the largest ratio of gain densities or their limiting value:
\begin{definition}
    Define the inequality-index of a distribution as
    \begin{equation}
        \Ib_\ell = \sup \qty(\frac{\ell(t_1)}{\ell(t_2)}) = \frac{\sup \ell(t_1)}{\inf \ell(t_2)}, \quad t_{1,2} \in [0,1] \ , \quad \ell(t_2) \neq 0 \ .
    \end{equation}
\end{definition}
The word ``inequality'' is used due to the original consideration of wealth distributions by Pareto, even though we do not attempt here to construct an inequality index that would fulfill the usual axioms from inequality economics. A natural theorem follows:
\begin{theorem}
    For all distributions satisfying the $p/(1-p)$-principle, the step function distribution is the most egalitarian with respect to $\Ib_\ell$.
\end{theorem}

\noindent \textbf{Proof:} Let $\ell$ satisfy the $p/(1-p)$-principle, and study its decreasing rearrangement $\ell^*$. By construction, for all $t_1 \in [0,p], \ t_2 \in [p,1], \ \ell^*(t_1) \geq \ell^*(t_2)$. By the second form of $\Ib_\ell$, we minimize inequality by minimizing the supremum and maximizing the infimum. 

Should the gain density on the subinterval $[0,p]$ differ from the uniform one, some differential intervals would have higher-than-uniform densities, while some others would have lower-than-uniform densities. This would raise the supremum. An analogous argument also shows this to lower the infimum on the subinterval $[p, 1]$. Hence, a uniform distribution for both subintervals yields the most egalitarian distribution, which together correspond to a step-function density.  \\

The normalized step function density satisfying the $p/(1-p)$-principle is
\begin{equation}
    \ell_p(t) = 
    \begin{cases}
        &\frac{1-p}{p}, \ t < p \\
        &\frac{p}{1-p}, \ t \geq p 
    \end{cases}\ .
\end{equation}
Hence, we can read off that the minimal inequality index for a given $p$ is $\frac{(1-p)^2}{p^2}$. This minimal inequality index as a function of $p$ is shown in Fig. \ref{fig:index}.
\begin{figure}[ht]
    \centering    
    \includegraphics[width=0.6\linewidth]{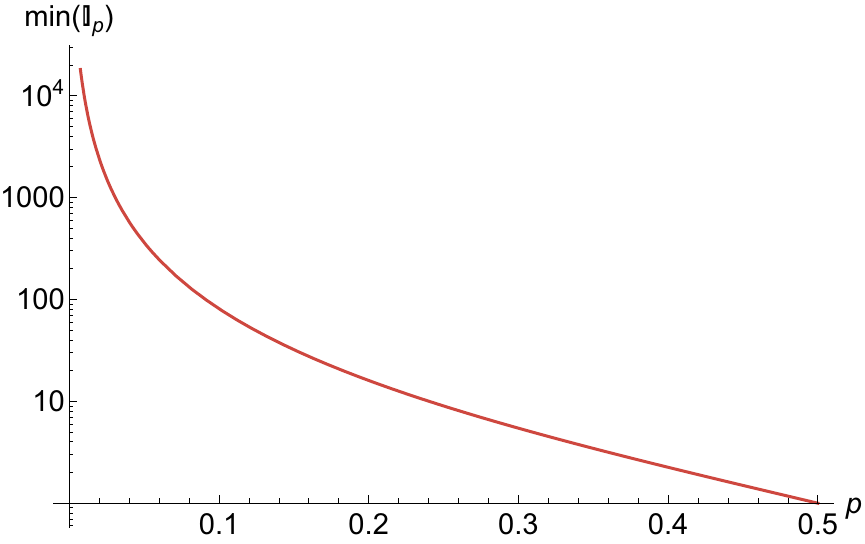}
    \caption{Minimal inequality index or ratio of gain densities a distribution must have to satisfy a given $p/(1-p)$-principle.}
    \label{fig:index}
\end{figure}

\subsubsection{Divergent profiles and rearrangements}

Consider a divergent rate of gains 
\begin{equation}
    \ell(t) = \frac{1}{2 \sqrt{2}} \frac{1}{\sqrt{\mid t - \frac{1}{2} \mid}} \ ,
\end{equation}
for which one must find the decreasing rearrangement. The distribution is symmetric with respect to $t = \frac{1}{2}$, increasing for $t < \frac{1}{2}$ and decreasing for $t > \frac{1}{2}$. In such a simple case, the decreasing rearrangement is achieved by shifting the right-hand side of the distribution to start from zero and by sending $t \to t/2$, so together, $t \to \frac{t}{2} + \frac{1}{2}$. This can be thought of as the continuous version of doubling the length of every bin.

Note that shifting the divergence to zero and re-normalizing is not in general equivalent to the decreasing rearrangement; as the rearrangement is done correctly, the density profile stays normalized. The rearranged distribution is
\begin{equation}
    \ell^*(t) = \ell\qty(\frac{t}{2}+\frac{1}{2}) = \frac{1}{2 \sqrt{t}} \ .
\end{equation}
The original and rearranged densities as well as their cumulative gain functions are shown in Fig. \ref{fig:div}. The equation for the generalized principle is $\sqrt{p} + p - 1 = 0$, which is solved by $p = (3  - \sqrt{5})/2 \approx 0.382$.

\begin{figure}[ht]
    \centering
    \includegraphics[width=\linewidth]{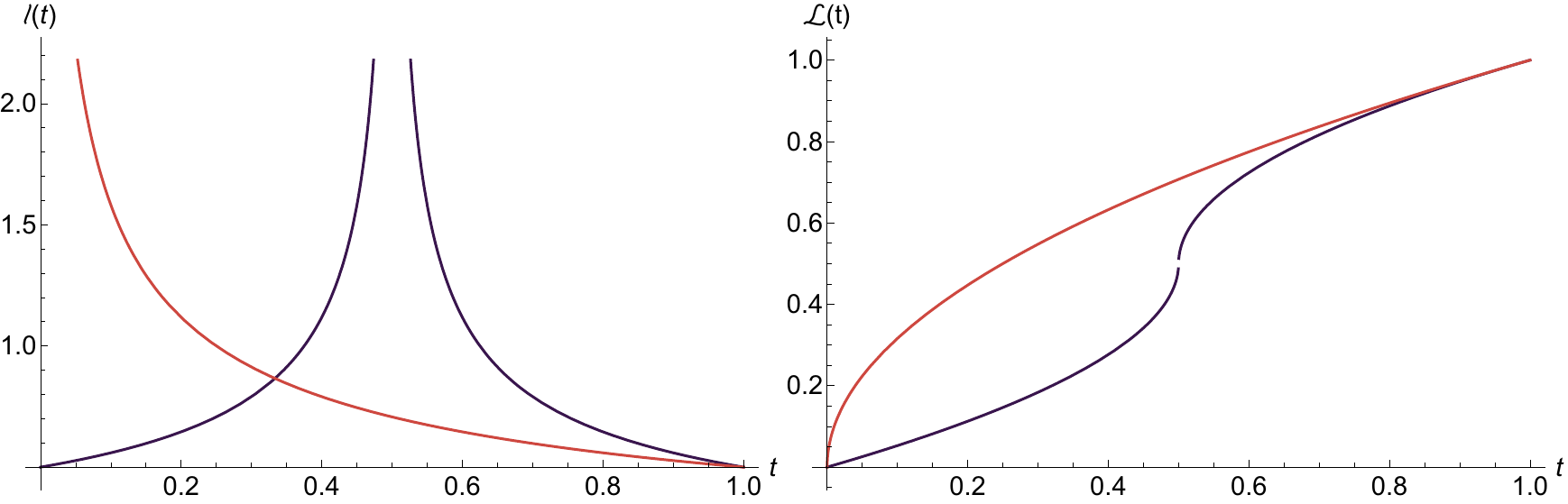}
    \caption{An $L^1$-integrable divergent density with its decreasing rearrangement, and their related cumulative gain functions. As must be, for all $t \in [0,1]$, $\L^*(t) \geq \L(t)$.}
    \label{fig:div}
\end{figure}

Remembering the requirement of no padding by zeros, define a distribution of periodic diminishing returns with essential support in $[0,1]$. To make finding the rearrangement analytically feasible, consider the following piecewise distribution:
\begin{equation}
    \ell(t) = 
    \begin{cases}
        \frac{16}{5}(1-4t), \ &0 \leq t <\frac{1}{4} \\
        \frac{12}{5}\qty(1-4\qty(t-\frac{1}{4})), \ &\frac{1}{4} \leq t < \frac{1}{2} \\
        \frac{8}{5}\qty(1-4\qty(t-\frac{1}{2})), \ &\frac{1}{2} \leq t < \frac{3}{4} \\
        \frac{4}{5}\qty(1-4\qty(t-\frac{3}{4})), \ &\frac{3}{4} \leq t \leq 1
    \end{cases} \ .
\end{equation}
The decreasing rearrangement is
\begin{equation}
     \ell^*(t) = 
     \begin{cases}
         \frac{64}{5}\qty(\frac{1}{4} - t), \ &0 \leq t < \frac{1}{16} \\
         \frac{192}{35}\qty(\frac{1}{2} - t), \ &\frac{1}{16} \leq t < \frac{5}{24} \\
         \frac{192}{65}\qty(\frac{3}{4} - t), \ &\frac{5}{24} \leq t < \frac{23}{48} \\
         \frac{192}{125}\qty(1 - t), \ &\frac{23}{48} \leq t \leq 1
     \end{cases} \ .
\end{equation}
Both of the above densities and their cumulative gain functions are shown in Fig. \ref{fig:per}. The generalized principle is satisfied with $p = \frac{11}{32}$.

\begin{figure}[ht]
    \centering
    \includegraphics[width=\linewidth]{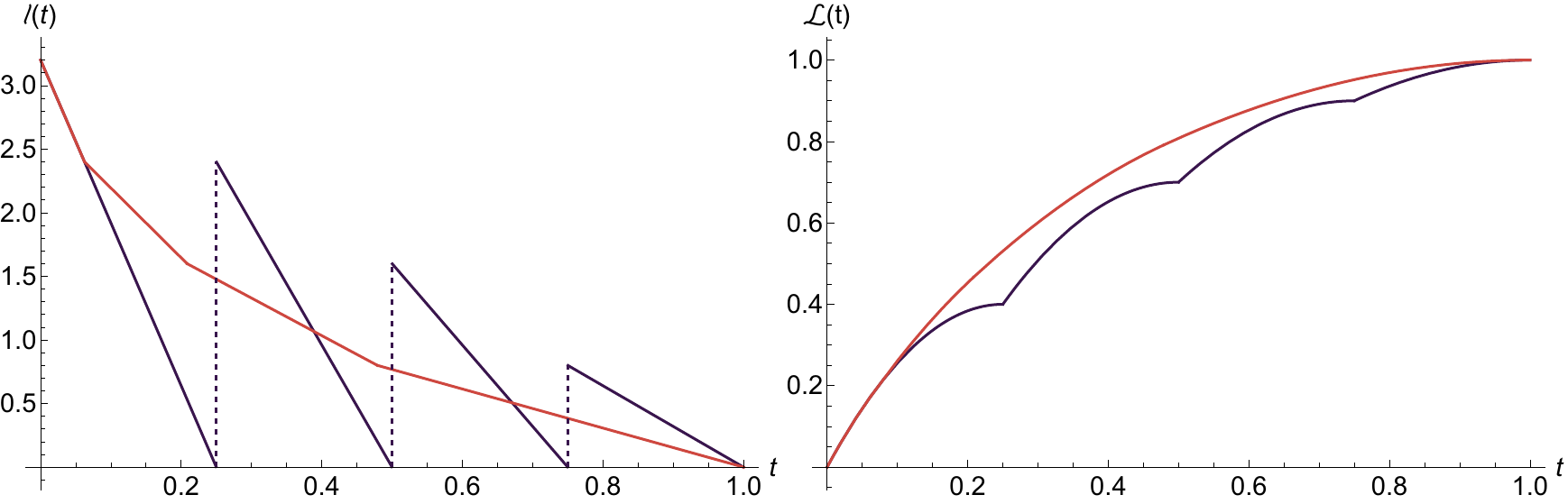}
    \caption{A periodic gain density with its decreasing rearrangement, and their related cumulative gain functions.}
    \label{fig:per}
\end{figure}

\subsubsection{Polynomial densities}
\label{sec:polynomialdensities}

As a final constructed example, consider the case of polynomial densities. Let the normalized, monotonically decreasing polynomial density with exponent $\alpha \geq 0$ be
\begin{equation}
    \ell_\text{pol}(t) = \frac{1-t^\alpha}{1 - \frac{1}{\alpha +1}} \ , \qquad \L_\text{pol}(t) = \frac{t(1 - t^\alpha + \alpha)}{\alpha} \ .
\end{equation}
Examples of such densities with $\alpha = \frac{1}{4}, 1, 3$ and $10$, and their cumulative gain functions are shown in Fig. \ref{fig:pol}. The generalized principles these distributions satisfy have $p \approx 0.339, 0.382, 0.434 $ and $0.476$, respectively.

\begin{figure}[ht]
    \centering
    \includegraphics[width=\linewidth]{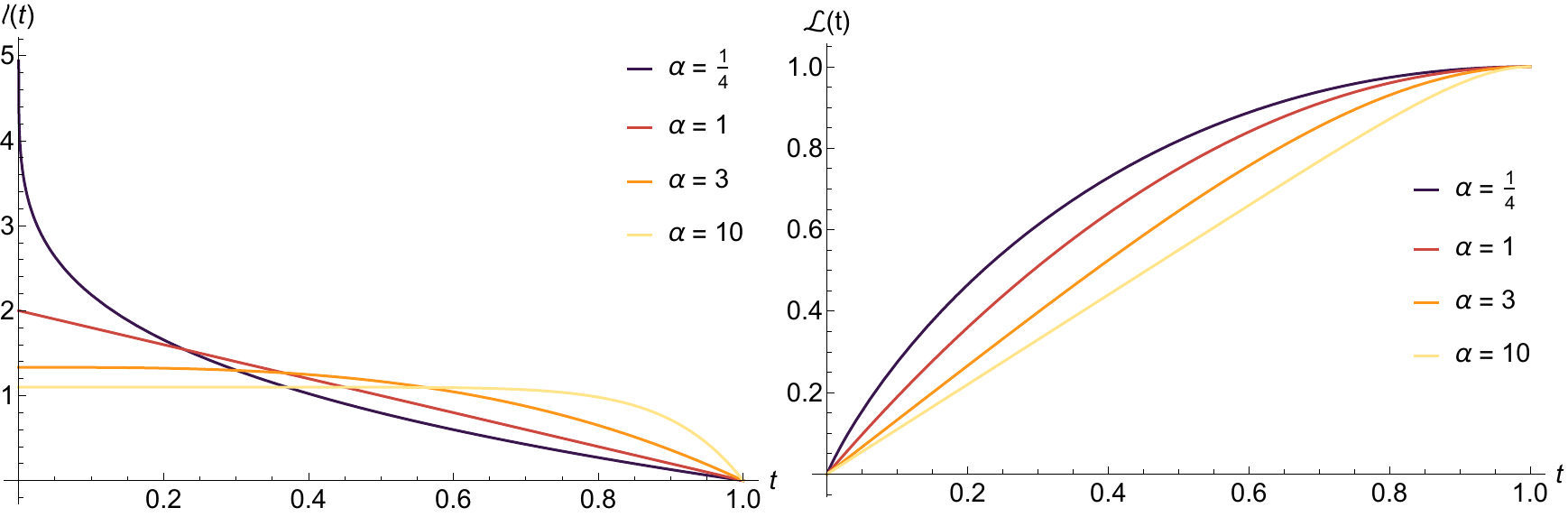}
    \caption{Polynomial densities with $\alpha = \frac{1}{4}, 1, 3$ and $10$, and their related cumulative gain functions.}
    \label{fig:pol}
\end{figure}
A natural question arises: can we always find an $\alpha \in [0, \infty)$ such that any given generalized principle is satisfied? That is, can we always find an $\alpha$ such that
\begin{equation}
\label{eq:polynomial}
    \frac{p(1-p^\alpha + \alpha)}{\alpha} = 1 - p \ ,
\end{equation}
for any $p \in (0, 1/2]$? On one hand, at the limit $\alpha \to \infty$ we obtain the uniform distribution, for which $p \to 1/2$. On the other hand, we may rearrange \eqref{eq:polynomial} into the form
\begin{equation}
    \frac{p^\alpha - 1}{\alpha} = 2 - \frac{1}{p} \ .
\end{equation}
We can take the limit $\alpha \to 0$ with $\lim_{\alpha \to 0} \frac{x^\alpha - 1}{\alpha} = \ln (x)$ for $x > 0$, following from the expansion $x^\alpha = \sum_{n = 0}^\infty \frac{(\alpha \log (x))^n}{n!}$:
\begin{equation}
\label{eq:logidentity}
    \frac{x^\alpha -1}{\alpha } = \sum_{n = 1}^\infty \frac{\alpha^{n-1} (\ln (x))^n}{n!} = \ln (x) + \frac{\alpha (\ln(x))^2}{2}+ \dots \to_{\alpha \to 0} = \ln (x) \ .
\end{equation}
Thus, the singularity at $\alpha = 0$ is removable, and we get a limiting equation for the polynomial profile,
\begin{equation}
    \ln (p) = 2 - \frac{1}{p} \ .
\end{equation}
Let $h(p) = \ln (p) - 2 + \frac{1}{p}$. Then $h'(p) = \frac{p-1}{p^2} < 0$, the function is strictly decreasing on $(0, 1/2]$, and by the IVT, there is a unique root of $h(p) = 0$. This can be found numerically to be $p^* \approx 0.318$. This is the lowest value of $p$ with which such a polynomial profile may satisfy the generalized principle. We will study similar limits below for common distribution families.

\subsection{Distribution families}
\label{sec:commondistribution}

Let us consider distribution families commonly encountered in datasets. In order not to limit ourselves to the parameter range $t \in [0,1]$, we must reparametrize the (truncated) distributions. This becomes especially apparent with the power-law distribution, in which the values are ill-defined with a bare variable value $0$. 

Define the value of the random variable $x \in [a,b]$ with an affine map
\begin{equation}
    x(t) = a + (b-a)t \ .
\end{equation}
We could of course consider other maps as well, but since the imbalance of gains is seen clearly on a linear axis, we shall use this parametrization. If the distribution on the original interval is $f_X$, the corresponding $\ell(t)$ is
\begin{equation}
    \ell(a,b;t) = f_X\qty(x(t)) x'(t) = f_X\qty(a + (b-a)\cdot t))\cdot(b-a) \ ,
\end{equation}
as well as the cumulative gain function,
\begin{equation}
    \L(a,b;t) = (b-a) \int_0^t \dd s \ f_X\qty(a + (b-a)\cdot s) \ .
\end{equation}

\subsubsection{Power-law distribution}
\label{sec:powerlaw}

Perhaps the first model that comes to mind in any application of the Pareto principle is that of the power-law, also known as a Pareto distribution. In fact, it is so commonly encountered that various works have focused solely on its form and interpretation across sciences \cite{Newman:2005, Clauset:2009}. Note however, that many datasets showcase a power-law \textit{tail} instead of a strict power-law for the full parameter range. For example in \cite{Clauset:2009} it was found that only one of the 24 datasets considered exhibited a statistically probable pure power-law for the full parameter range.

The distribution of a bounded power-law with exponent or \textit{scale} $\alpha$ and lower and upper limit $a$ and $b$ is given by
\begin{equation}
    f_\text{$X$, pow}(a, b, \alpha; x) = 
    \begin{cases}
        &\frac{1-\alpha}{b^{1-\alpha}- a^{1-\alpha}}x^{-\alpha} \ , \ \alpha \neq 1 \\
        &\frac{1}{\log(b/a)} \frac{1}{x} \ , \ \alpha = 1 
    \end{cases} \ .
\end{equation}
We have $x > 0$ and assume $\alpha > 1$. To be explicit, the transformed gain density and cumulative gain function with $\alpha > 1$ are given by
\begin{equation}
\label{eq:powerlawabalpha}
    \begin{split}
        \ell_\text{pow}(a,b,\alpha;t) &=  \frac{(1-\alpha)(b-a)}{b^{1-\alpha}- a^{1-\alpha}}\qty(a + (b-a)t)^{-\alpha}  \\
        \L_\text{pow}(a,b,\alpha;t) &= \frac{\qty(a + (b-a)t)^{1-\alpha} - a^{1-\alpha}}{b^{1-\alpha} - a^{1-\alpha}} \ .
    \end{split}
\end{equation}
We can write both distributions in \eqref{eq:powerlawabalpha} using the ratio $r = b/a$ as
\begin{equation}
    \begin{split}
        &\ell_\text{pow}(r,\alpha;t) = \frac{(1-\alpha)(r - 1)}{r^{1-\alpha}-1} \qty(1 + (r - 1)t)^{-\alpha}  \\
        &\L_\text{pow}(r,\alpha;t) = \frac{(1 + (r-1)t)^{1-\alpha}-1}{r^{1-\alpha}-1} \ .
    \end{split}
\end{equation}
Gain densities with ratio and scale combinations $(r,\alpha) = (2,2), (2,3), (10,2)$ and $(10,3)$, as well as their corresponding cumulative gain functions are shown in Fig. \ref{fig:pow}. The principles these profiles satisfy have $p \approx 0.414, 0.373, 0.240$ and $0.160$, respectively.
\begin{figure}
    \centering
    \includegraphics[width=\linewidth]{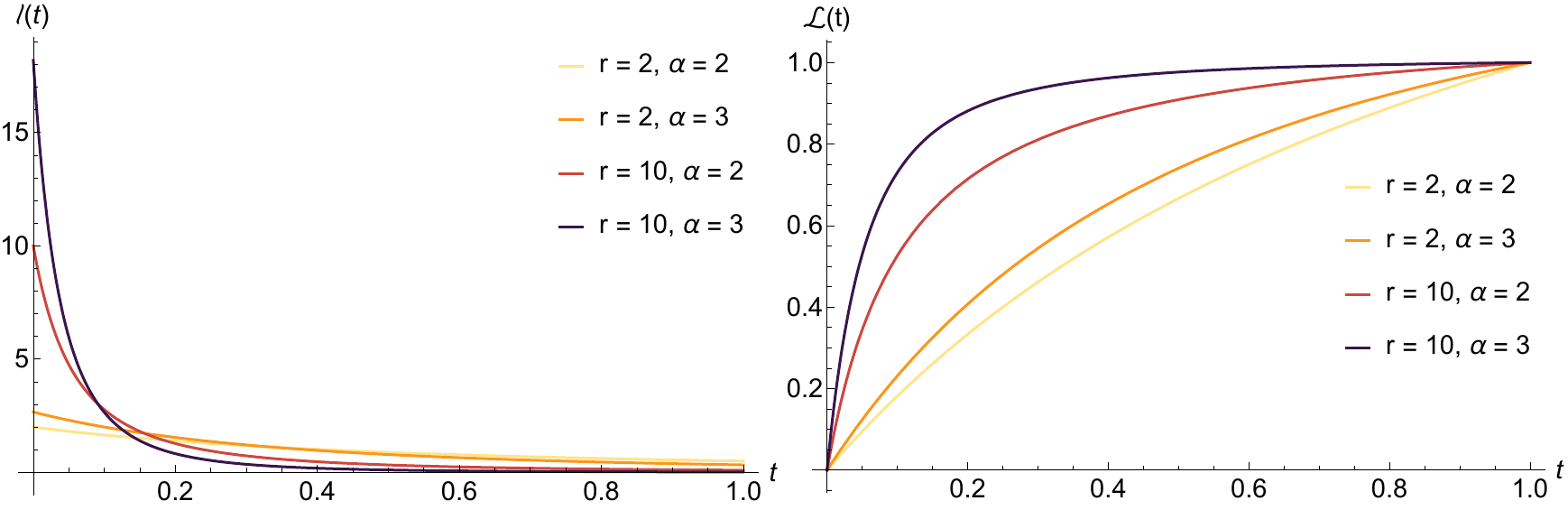}
    \caption{Power-law densities with ratio and scale combinations $(r,\alpha) = (2,2), (2,3), (10,2)$ and $(10,3)$, and their related cumulative gain functions.}
    \label{fig:pow}
\end{figure}

To find which generalized principles can be satisfied by tuning $r$ and $\alpha$, study the zeros of
\begin{equation}
\label{eq:powerlawparetopoint}
    h(r, \alpha; p) =\frac{(1+ (r-1)p)^{1 - \alpha} - 1}{r^{1-\alpha}-1} - 1 + p \ .
\end{equation}
First, consider fixing $\alpha$ and varying $r$. The singularity at $r = 1$ is removable, and the two limits give
\begin{equation}
    \lim_{r \to 1}h(r;\alpha, p) = -1 + 2p \leq 0 \ , \qquad \lim_{r \to \infty}h(r,\alpha;p) = p > 0 \ .
\end{equation}
The function $h$ is continuous on the interval, so by IVT, there exists a root for some $r \in [1,\infty)$, independent of the value of $\alpha$. Hence, if $r$ is allowed to vary, any generalized principle can be satisfied. 

Now consider fixing $r$ and varying $\alpha$. Let $q(r,p) = 1 + (r-1)p, \ 1 < q \leq r$, with which $h$ becomes
\begin{equation}
    h(r,\alpha;p) = \frac{q(r,p)^{1-\alpha}-1}{r^{1-\alpha}-1} -1 + p\ .
\end{equation}
The limit $\alpha \to 1$ is equivalent to $\beta = \alpha - 1 \to 0$, the singularity of which is again removable. Using \eqref{eq:logidentity} we get
\begin{equation}
    \lim_{\alpha \to 1} h(r,\alpha;p) = \frac{\ln q(r,p)}{\ln r} - 1+ p \ , \qquad \lim_{\alpha \to \infty} h(r,\alpha;p) = p \ .
\end{equation}
The endpoint is again positive, but the first limit requires further study:
\begin{equation}
    \pdv{}{\beta}\qty(\frac{q(r,p)^{-\beta} - 1}{r^{-\beta} - 1}) = \qty(\frac{r}{q(r,p)})^\beta\frac{1}{\qty(r^\beta - 1)^2}\qty[\qty(r^\beta - 1)\ln (q(r,p)) - \qty(q(r,p)^\beta - 1)\ln (r)] \ .
\end{equation}
The prefactor is always positive, since $\beta > 0$ and $r > 1$. Let the sign function be $S(\beta) = (r^\beta - 1)\ln (q(r,p)) - (q(r,p)^\beta - 1)\ln (r)$. Then $\lim_{\beta \to 0} S(\beta) = 0$ and 
\begin{equation}
    S'(\beta) = \ln(r) \ln(q(r,p))\qty(r^\beta - q(r,p)^\beta) \ ,
\end{equation}
and since $r \geq q(r,p) > 1$ and $\beta > 0$, $S'(\beta) \geq 0$. Hence, $h(r,\alpha;p)$ is monotonous, and the limit of $\alpha \to 1$ gives the highest possible value with which the principle can be satisfied. The condition for the existence of a solution with a given $p$ reads
\begin{equation}
    \frac{\ln(1 + (r-1)p)}{\ln(r)} < 1 - p \ .
\end{equation}
By monotonicity of the logarithm, the constraint can be written as
\begin{equation}
    1 + (r-1)p < r^{1-p} \ .
\end{equation}
Therefore, for given $r$ and $p$, the generalized Pareto principle can be satisfied by tuning $\alpha$ if and only if this condition is satisfied. Stated otherwise, for a given $r$, there exists a maximal $p$ for which the generalized principle can be satisfied by tuning $\alpha$ down. By the limit $\alpha \to \infty$, we obtain an increasingly spiked distribution, allowing us to satisfy the generalized principle with $p$ arbitrarily close to 0. We numerically solve and plot $p_\text{max}(r)$ in Fig.~\ref{fig:pmax}. In particular, it is found that satisfying the canonical 0.2/0.8-principle with power-law distributed data becomes impossible for any $\alpha$ after $r \approx 3100$.

\begin{figure}[ht]
    \centering
    \includegraphics[width=0.6\linewidth]{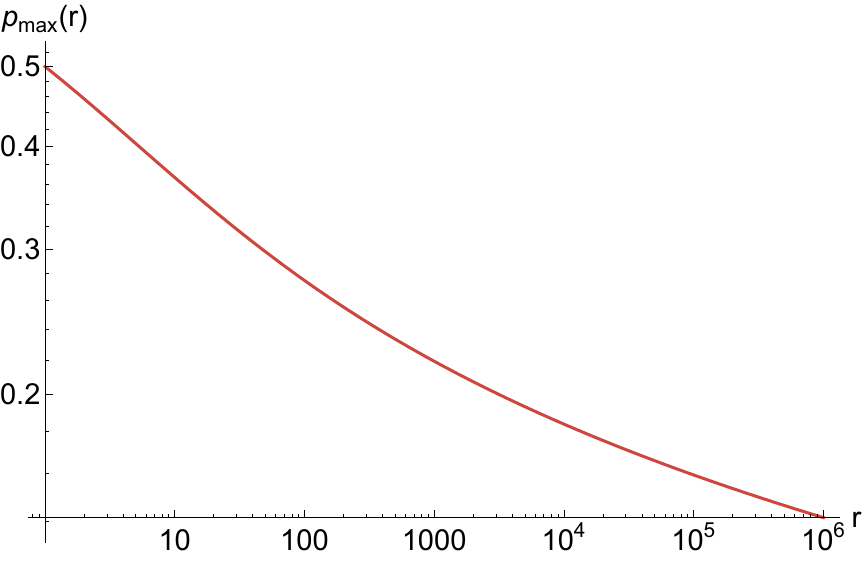}
    \caption{The maximal value of $p$ with which the generalized Pareto principle can be satisfied by varying $\alpha$ with a fixed $r$. The canonical 0.2/0.8-principle becomes impossible with a pure power-law after $r \approx 3100$.}
    \label{fig:pmax}
\end{figure}

\subsubsection{Exponential distribution}
\label{sec:exponentialdistributions}

Any process where decay has the same probability of occurring at any moment in time follows an exponential law. For a given \textit{rate} $\lambda > 0$, the truncated exponential distribution and its cumulative distribution are given by
\begin{equation}
\begin{split}
    \ell_\text{exp}(a, b, \lambda;t) &= (b-a) \frac{\lambda e^{-\lambda \qty(a + (b-a)t)}}{e^{-\lambda a} - e^{-\lambda b}} \\
    \L_\text{exp}(a, b,\lambda; t) &= \frac{e^{-\lambda a}- e^{-\lambda \qty(a + (b-a)t)}}{e^{-\lambda a} - e^{-\lambda b}} \ .
\end{split}
\end{equation}
The shape of both distributions is defined by the combined parameter $\lambda (b-a) = \Lambda > 0$ as
\begin{equation}
    \begin{split}
        \ell_\text{exp}(\Lambda; t) &=  \frac{\Lambda e^{-\Lambda t}}{1 - e^{-\Lambda}} \  \\
        \L_\text{exp}(\Lambda; t) &= \frac{1 - e^{-\Lambda t}}{1 - e^{-\Lambda}} \ .
    \end{split}
\end{equation}
Example gain densities with $\Lambda = 0.1, 1, 4$ and 10, and their related cumulative gain functions are shown in Fig. \ref{fig:exp}. The principles these profiles satisfy have $p \approx 0.494, 0.438, 0.295$ and $0.175$, respectively. 
\begin{figure}
    \centering
    \includegraphics[width=\linewidth]{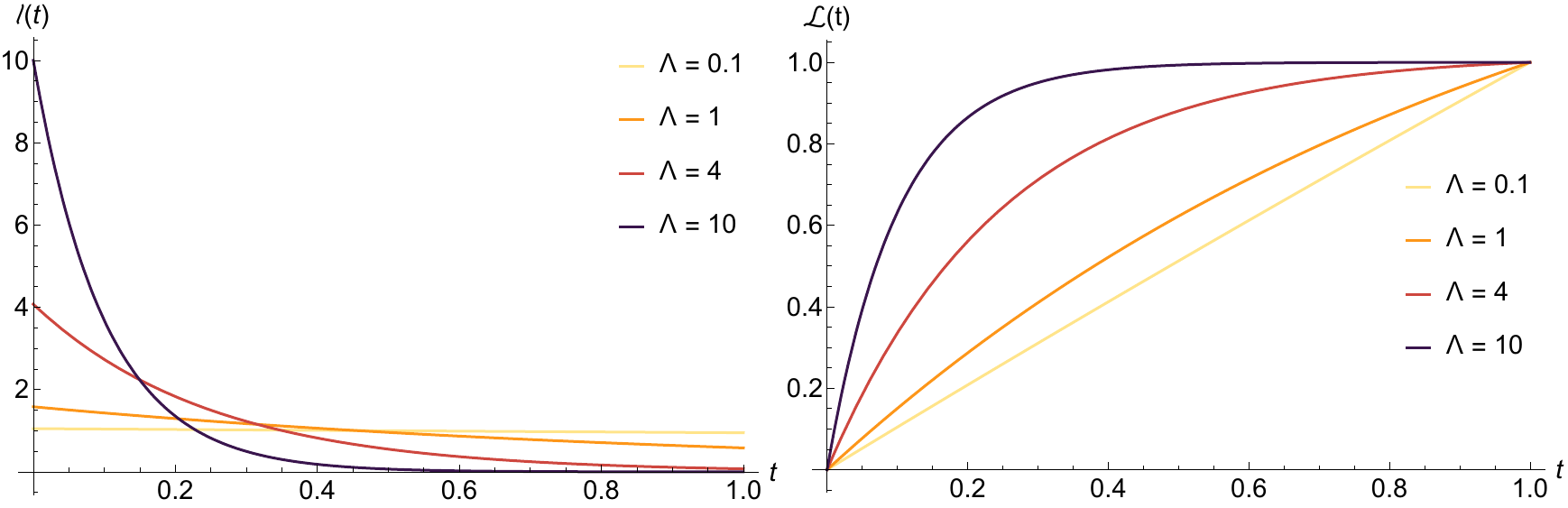}
    \caption{Exponential densities with $\Lambda = 0.1, 1, 4$ and 10, and their related cumulative gain functions.}
    \label{fig:exp}
\end{figure}

The question about the existence of any generalized principle is straightforward. Define
\begin{equation}
\label{eq:exponentialparetopoint}
    h(\Lambda; p) = \frac{1- e^{-\Lambda p}}{1 - e^{-\Lambda}} - 1  + p \ .
\end{equation}
The singularity at $\Lambda = 0$ is removable, and for all $\Lambda > 0, \ h$ is continuous. The limits are
\begin{equation}
    \lim_{\Lambda \to 0}h(\Lambda; p) = -1 + 2p \leq 0\ , \qquad \lim_{\Lambda \to \infty}h(\Lambda; p) = p > 0\ ,
\end{equation}
and by IVT, there exists at least one value of $\Lambda$ with which any generalized principle can be satisfied.

\subsubsection{Normal distribution}
\label{sec:normaldistributions}

Finally, the normal distribution. In principle, we need both the mean $\mu$ and standard deviation $\sigma$ to describe a truncated normal distribution on some interval $[a,b]$. However, we can often choose the distribution to be symmetrically truncated with equally many standard deviations on both sides. In this case $\mu = \frac{a+b}{2}$, making the density
\begin{equation}
    f_{X, \ \text{norm}}(x) = \frac{1}{\sqrt{\pi} \sigma \erf\qty(\frac{b-a}{2 \sigma})}\exp \qty(- \frac{\qty(x - \frac{a+b}{2})^2}{\sigma^2}) \ .
\end{equation}
As before, use the affine map and perform the decreasing rearrangement by sending $t \to \frac{t}{2} + \frac{1}{2}$. The density becomes
\begin{equation}
\ell(a,b,\sigma;t) = \frac{b-a}{\sqrt{\pi} \sigma \erf \qty(\frac{b-a}{2\sigma})}\exp \qty(- \frac{(b-a)^2 t^2}{4 \sigma^2}) \ .
\end{equation}
Similarly to the exponential distribution, set $\Sigma = \frac{b-a}{2 \sigma}$ to get
\begin{equation}
\begin{split}
    \ell_\text{norm}(\Sigma; t) &= \frac{2 \Sigma}{\sqrt{\pi} \text{erf}\qty(\Sigma)} \exp \qty(- \Sigma^2 t^2) \\
    \L_\text{norm}(\Sigma; t) &= \frac{\text{erf}\qty(\Sigma t)}{\text{erf} \qty(\Sigma)} \ .
\end{split}
\end{equation}
Example densities with $\Sigma = 1, 3, 5$ and 10, and their corresponding cumulative profiles are shown in Fig. \ref{fig:norm}. The principles these profiles satisfy have $p \approx 0.443, 0.264, 0.187$ and $0.112$, respectively.
\begin{figure}[ht]

    \centering
    \includegraphics[width=\linewidth]{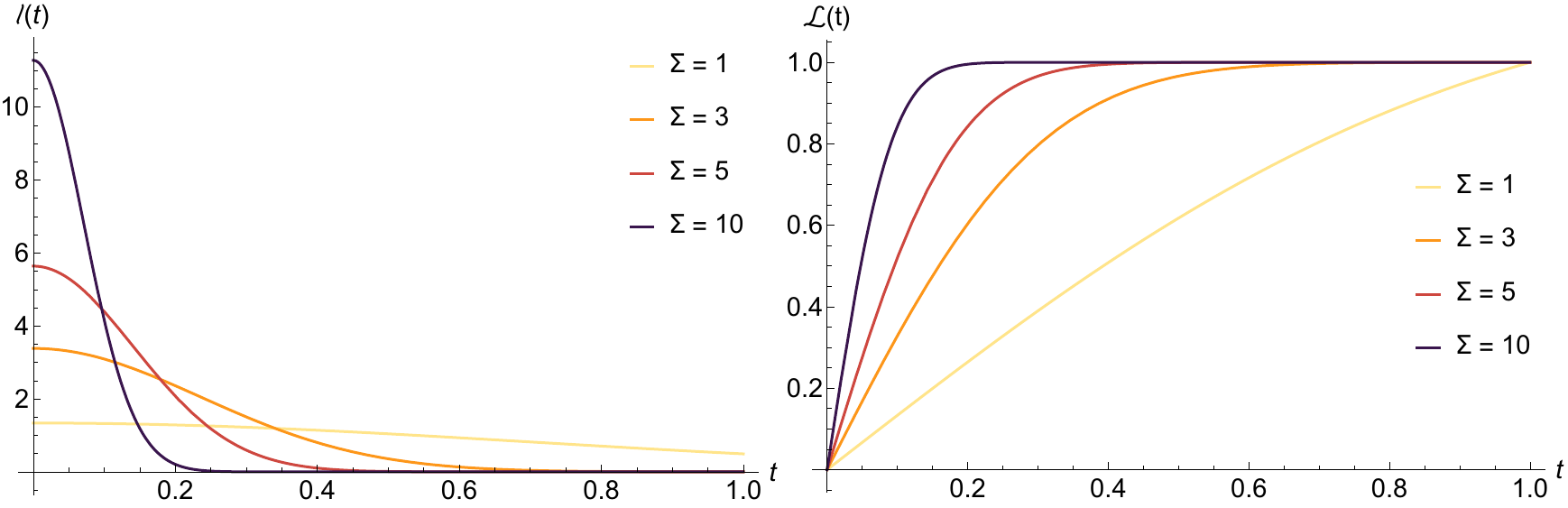}
    \caption{Normal densities  with $\Sigma = 1, 3, 5$ and 10, and their related cumulative gain functions.}
    \label{fig:norm}
\end{figure}

Analogous to previous cases, set
\begin{equation}
\label{eq:normalparetopoint}
    h(\Sigma; p) = \frac{\erf(\Sigma p)}{\erf(\Sigma)} - 1 + p \ .
\end{equation}
The singularity at $\Sigma = 0$ is removable, $h$ is continuous for all $\Sigma > 0$ and again
\begin{equation}
    \lim_{\Sigma \to 0} h(\Sigma; p) = -1 + 2p \leq 0 \ , \qquad \lim_{\Sigma \to \infty} h(\Sigma;p) = p > 0 \ ,
\end{equation}
so by IVT, there exists a value of $\Sigma$ with which any generalized principle can be satisfied.

\section{Common generalized principles and social discussion}
\label{sec:discussion}

With the generalized Pareto principle formalized and explicit density functions analyzed, we now turn to two questions. First, given common distributions and realistic parameter ranges, which $p/(1-p)$-principles should one expect to observe in practice? Second, if such asymmetries are in large part structural, how should we interpret their role in discussions of inequality, decision-making, and management?

Section \ref{sec:canonicalprinciples} addresses the first question by combining the results of Section \ref{sec:examplesexistence} with empirical parameter ranges and typical dataset sizes. Section \ref{sec:socialimplications} then discusses the broader social and conceptual implications.

\subsection{Common generalized Pareto principles}
\label{sec:canonicalprinciples}

Section \ref{sec:examplesexistence} does not attempt to claim that every dataset would satisfy every generalized Pareto principle. With freedom in binning, a dataset may satisfy a set of generalized principles around some value of $p$, but by no means all of them. The previous examples show that for some of the commonly encountered distribution families -- power-laws, exponentials, and normals -- one cannot \textit{a priori} rule out many or any of the generalized principles that a dataset following such distribution can satisfy. One must study further, which values of $p$ are most likely to arise in data.

On one hand, the most direct approach would be to take a look at a representative sample of phenomena showing behavior according to these distributions and computing the generalized Pareto principles they satisfy. Of course, in finite samples, one can observe apparent manifestations of any version of the generalized principle -- the smaller the dataset, the larger the role played by fluctuations. The underlying generating distribution dictates only the $p$ around which such fluctuations can be expected to occur.

On the other hand, \emph{if} a dataset can be assumed to follow a given distribution, we may infer how much range the values in data likely cover -- effectively, we can give a robust estimate for $b/a$ or $b-a$ as a function of the sample size $N$. In essence, this method utilizes the argument of no padding by zeros, discussed in detail in Appendix \ref{app:limitationsoff}: we truncate the distribution as there likely are no more data points.

Independent of either of the above tests, extreme generalized principles very close to $p = 1/2$ or $p \to 0$ require either highly concentrated or almost perfectly uniform gain profiles, and should therefore be expected to be much less common than intermediate cases in data. This in itself gives a reason to believe that the canonical Pareto principle should be more common than many other cases, with the $0.25/0.75$-principle being the most common if symmetricity with respect to extreme cases is assumed. The quantitative calculations below make this intuition more precise for the three distribution families considered.

\subsubsection{Power-law distribution}

We study the power-law parametrized distributions across scientific disciplines from \cite{Clauset:2009} which give a statistically good or moderate fit with a pure power-law in Tables 6.2 and 6.3 of \cite{Clauset:2009}. The ratio $r$ is computed by us from $x_\text{max}$ and the fitted $\hat x_\text{min}$ of the original dataset, $\hat \alpha$ was predicted with maximum likelihood estimation by the authors of \cite{Clauset:2009}, and the generalized principle $p$ is solved numerically by us from the equation \eqref{eq:powerlawparetopoint}. 

We exclude datasets in which the power-law is not a plausible model over the full parameter range, based on major deviations from linearity in the log-log plots of Fig. 6.1 of \cite{Clauset:2009}. The remaining datasets, their parameters given in Table 6.1 of \cite{Clauset:2009} and the $p$ they satisfy the generalized principle with are listed in Table \ref{tab:powerlaws}.

\begin{table}[ht]
    \centering
    \begin{tabular}{lccccc}
    \toprule
    dataset & $\hat x_\text{min}$ & $x_\text{max}$ & $r$ & $\hat \alpha$ & $p$ \\
    \midrule
    sales of books ($\times 10^3$)      & 2400 & 19077 & $7.95$ & $3.7$  & $0.147$  \\
    religious followers ($\times 10^6$) & 3.85 & 1050  & $273$  & $1.8$  & $0.0755$ \\
    intensity of wars                   & 2.1  & 382   & $182$  & $1.7$  & $0.102$  \\
    protein interaction degree          & 5    & 56    & $11.2$ & $3.1$  & $0.144$  \\
    terrorist attack severity           & 12   & 2749  & $229$  & $2.4$  & $0.0394$ \\
    count of word use                   & 7    & 14086 & $2010$ & $1.95$ & $0.0240$ \\
    \bottomrule
    \end{tabular}
    \caption{Statistically possible full parameter range power-law behavior for phenomena present in the datasets analyzed in \cite{Clauset:2009}. The generalized principles satisfied by the datasets are solved from \eqref{eq:powerlawparetopoint}.}
    \label{tab:powerlaws}
\end{table}

For these six cases, we find generalized principles with $p \in [0.024, 0.15]$, all more unequal than the original Pareto principle. The general trend is that the higher the ratio $r$, the lower $p$ the principle is satisfied with, since $\alpha$'s are all approximately in the same range. Intuitively, the more orders a pure power-law behavior covers, the more extreme the resulting differences in gains. 

Based on the values presented in Table \ref{tab:powerlaws} and \cite{Newman:2005, Clauset:2009}, we may take as a rough estimate that across various phenomena, common ranges for parameter values are $r \in [10, 10^3]$ and $\alpha \in [1.5, 3.5]$. We numerically solve $p$ on this range and plot $p$ as a function of $r$ and $\alpha$ in Fig. \ref{fig:commonpowerlaws}.

\begin{figure}[ht]
    \centering
    \includegraphics[width=0.7\linewidth]{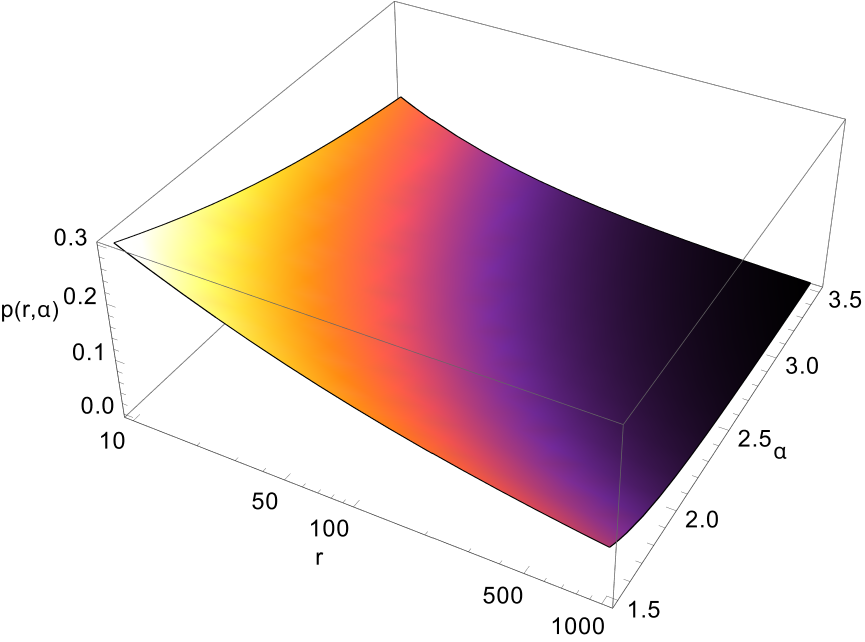}
    \caption{The generalized Pareto principles satisfied by truncated power-laws on common parameter ranges for $r$ and $\alpha$.}
    \label{fig:commonpowerlaws}
\end{figure}
We could also estimate the relevant parameter range as a function of dataset size. Effectively, we want to truncate the distribution or ``cut off its tail'' as no more data points are likely to lie outside of it. This requires us to approximate what is the maximal value found with $N$ data points. A simple estimate for the maximal value with a non-truncated power-law is \cite{Newman:2005}
\begin{equation}
\label{eq:powerlawestimate}
    \mathbb{E}(x_\text{max}) \sim N^{1/(\alpha-1)} \ .
\end{equation}
Even for exponents in the empirically common range, $\Eb(x_\text{max})$ grows very rapidly with $N$, and the largest observations can be several orders of magnitude larger than the smaller ones. This heaviness of the tail of a power-law may often be taken synonymous with the fact that the distribution does not have exponential suppression for outliers. Truncated power-laws hence tend to generate stronger deviations from equality and generalized principles with smaller $p$ than many other distributions with exponential decay of the tail(s).

\subsubsection{Exponential distribution}

For exponential distributions, the dataset size argument can be made particularly explicit, since the shape is governed by a single parameter. As seen, a truncated exponential is governed solely by $\Lambda = (b-a)\lambda$, which describes how many inverse decay lengths are included in the observation window.

Assuming we have $N$ data points distributed according to an \textit{a priori} non-truncated exponential, we may infer the expected minimal and maximal values of the sample. This tells us immediately at what point we must truncate the distribution, and so, $\Lambda$ as a function of $N$. 

For size $N$ independent and identically distributed samples from Exp$(\lambda)$, also their differences are exponentially distributed (see \textit{e.g.} \cite{Arnold:2008}). This allows us to compute the expected minimal and maximal values in the dataset as
\begin{equation}
    \mathbb{E}[\min (X_N)] = \frac{1}{N \lambda} \ , \qquad \mathbb{E}[\max (X_N)] = \frac{1}{\lambda} \sum_{j=1}^N \frac{1}{j} = \frac{H_N}{\lambda} \ ,
\end{equation}
where $H_N$ is the $N^{\text{th}}$ harmonic number. Since $N \gg 1$, we can use the asymptotic expansion
\begin{equation}
    H_N = \log(N) + \gamma + O(1/N) \ ,
\end{equation}
where $\gamma$ is the Euler-Mascheroni constant.

With a lower limit on observations, the expected minimal value in a dataset of size $N$ is $a + \frac{1}{N \lambda} = a + O(1/N)$, while the highest is $b = (\log(N) + \gamma + O(1/N))/\lambda + a$. Therefore, we get an estimate for $\Lambda$ in terms of $N$,
\begin{equation}
    \Lambda = \lambda (b-a) \approx \log(N) + \gamma \ .
\end{equation}
Thus, satisfying a given generalized Pareto principle with an exponentially distributed dataset is effectively just a question about dataset size. The more observations you have, the higher the probability that some of them are found in the extreme tail, lowering the value of $p$.

Finally, can we estimate $N$ any better than $\Lambda$? We can at least give plausible ranges of $N$, depending on real studies conducted. Looking at \cite{Clauset:2009} and \cite{Limpert:2001} purely from the point of view of the sizes of datasets, we shall suppose that the size of a dataset commonly encountered in empirical studies lies roughly between $N = 10^2$ and $N = 10^5$. If similar-sized datasets are also used to describe exponentially distributed data, we get an estimate of the range of principles satisfied:
\begin{equation}
    \begin{cases}
        & N = 10^2 \\
        & N = 10^3 \\
        & N = 10^4 \\
        & N = 10^5 
    \end{cases}
    \implies 
    \begin{cases}
        &\Lambda \approx 5.18 \\
        &\Lambda \approx 7.48 \\
        &\Lambda \approx 9.79 \\
        &\Lambda \approx 12.1
    \end{cases}
    \implies
    \begin{cases}
        &p \approx 0.258 \\
        &p \approx 0.209 \\
        &p \approx 0.177 \\
        &p \approx 0.154
    \end{cases} \ .
\end{equation}
Thus, any dataset with a number of data points $N \in [10^2, 10^5]$ following a truncated exponential is likely to follow a generalized Pareto principle that is close to the canonical 0.2/0.8-principle. Naturally, the realized principle of a sample fluctuates around these expected values, with fluctuations playing a larger role as $N$ decreases.

\subsubsection{Normal distribution}

Due to the central limit theorem, the normal distribution is likely the most commonly encountered distribution across empirical data, and understanding the principle for normal distributions covers a wide range of phenomena across all sciences.\footnote{Also log-normals are encountered across sciences \cite{Limpert:2001}, but naturally, log-normals and normals are connected; in a normal distribution, changes are additive, while in a log-normal, they are multiplicative. If $X$ is log-normally distributed, $\log(X)$ is normally distributed. We could in principle study also log-normals by suitable changes of parameters, but here we focus only on the normal, mainly due to its prevalence and symmetry properties.}

We have considered the normal distributions to be truncated symmetrically with respect to the mean $\mu = \frac{a+b}{2}$. Thus, the parameter interval is $[a,b] = [\mu - k \sigma, \mu + k \sigma]$ with $k$ describing how many standard deviations of observations are considered, and in fact,
\begin{equation}
    \Sigma = \frac{2 k \sigma}{2 \sigma} = k \ .
\end{equation}
With a size $N$ dataset, we want to find $\Sigma_N$ such that approximately one observation is left in the cut tail. This gives a condition
\begin{equation}
    2 (1 - \Phi(\Sigma_N)) \approx \frac{1}{N} \ ,
\end{equation}
where $\Phi(x) = \frac{1}{2}\qty(1 + \erf\qty(x/\sqrt{2}))$. Inverting this relation, we find
\begin{equation}
    \Sigma_N = \sqrt{2}\erf^{-1}\qty(1- \frac{1}{N}) \ .
\end{equation}
We again consider usual dataset sizes to be in the range $N \in [10^2, 10^5]$ and find
\begin{equation}
    \begin{cases}
        & N = 10^2 \\
        & N = 10^3 \\
        & N = 10^4 \\
        & N = 10^5 
    \end{cases}
    \implies 
    \begin{cases}
        &\Sigma \approx 2.58 \\
        &\Sigma \approx 3.29 \\
        &\Sigma \approx 3.89 \\
        &\Sigma \approx 4.42
    \end{cases}
    \implies
    \begin{cases}
        &p \approx 0.290 \\
        &p \approx 0.248 \\
        &p \approx 0.222 \\
        &p \approx 0.204
    \end{cases} \ .
\end{equation}
Hence, for any dataset with a number of data points in the common range following a truncated normal, the generalized principles can be expected to fall approximately between $p \in [0.20, 0.29]$, again with sample variance kept in mind. We could therefore expect generalized Pareto principles close to the original 0.2/0.8-principle to be common also in settings where the dataset is well approximated by a truncated normal.

Together with the exponential case, this helps to explain why empirical Pareto-like statements often cluster around values similar to 0.2/0.8, even though the exact canonical ratio has no mathematically special status.

In the language of the Kolkata index, the above results for the exponential predict a range of $k_F \in [0.74, 0.85]$ while the normal predicts $k_F \in [0.71, 0.80]$ across typical sample sizes. In \cite{Ghosh:2021} it was conjectured, based on citation data and a polynomial parametrization of the Lorenz curve, that the Kolkata index saturates near $k_F \approx 0.86$, proposing this as a possible universal social constant. 

The present predictions for exponential and normal families sit strictly below this conjectured saturation across the full range of $N$ considered. However as we saw, power law distributions can produce substantially smaller $p^*$, that is, larger $k_F$; in the tail-heavy regime some of the cases in Table~\ref{tab:powerlaws} approach or exceed the conjectured saturation. These observations suggest that the $0.86$ saturation could be characteristic only of heavy-tailed distributions, whereas the exponential and normal families at typical sample sizes represent a less extreme baseline.

\subsection{Discussion}
\label{sec:socialimplications}

The framework, its general existence results, and the distribution-specific analyses above suggest that the (generalized) Pareto principle may primarily be a structural property of gain distributions, rather than a rare empirical finding. The principle is frequently used to motivate inequality policies, operational choices, and management heuristics. This raises the question: how should we interpret and use the principle if it is, to a significant extent, expected to arise structurally?

Since some $p/(1-p)$-type balance is guaranteed, the question becomes which distributional form underlies the data, and which value of $p$ can we expect and how to evaluate it. For normal and exponential distributions with realistic sample sizes, we have seen that $p$ often lies close to 0.2, whereas for power-laws the heavier tails naturally produce much smaller values of $p$, \textit{i.e.} more extreme concentration.

These distribution families have mostly well-understood generating mechanisms. Normal distributions are linked to the central limit theorem, while exponentials emerge when decay-like events have a constant probability of taking place. For power-laws, various dynamical mechanisms have been proposed; a prominent class is preferential attachment, formalizing the ``rich get richer''-effect, which naturally generates heavy tails and has been studied for example in networks \cite{Barabasi:1999, Newman:2005}.

In decision-based processes, a possible lower-level explanation for heavy-tailed outcomes is the prevalence of bounded rationality in the sense of Simon \cite{Simon:1955}. When agents face many alternatives, systematically favoring familiar or salient options (\textit{e.g.} known brands, heavily advertised products) is a natural way to reduce cognitive load. Such behavior can effectively implement preferential attachment: options that are already popular attract disproportionate attention, reinforcing heavy-tailed distributions of outcomes.

From the perspective of inequality and fairness, the mere presence of concentration is not by itself a meaningful diagnostic. One should be interested in which distribution is followed and what is an acceptable value of $p$. Internal dynamics of the distribution also play an important role: the generalized principle takes no stance on who occupies which fraction, and it is a highly non-trivial question whether inequality is stratified or allows for mobility and internal transitions.

The results also suggest that appealing to Pareto-type imbalances as a justification of inequality on efficiency grounds (``the high-output fraction deserves more because they produce more'') should be treated with care. Since some form of $p/(1-p)$-imbalance is guaranteed in bounded cumulative processes, the existence of such an imbalance alone cannot resolve normative questions about fairness. This descriptive-versus-prescriptive distinction is the foundational concern of welfare-economic inequality measurement \cite{Atkinson:1970, Sen:1973}, and the present commentary is in agreement with that tradition. Likewise in management and decision-making, the Pareto principle has been widely popularized as advice to ``focus on the vital few''. This can be useful when resources are limited and the underlying distribution is indeed highly skewed. However, the consequences of such heuristics depend sensitively on the actual distributional form: in some regimes, a strong focus on the top fraction may overlook structural causes of unequal outputs or under-utilize a short but heavy tail of contributors. 

Over-reliance on the 0.2/0.8-rule can also make diagnosis shortsighted; there is no mathematical reason to single out the canonical principle, since the specific value $p = 0.2$ plays no distinguished role in the formalism. Rather, it seems to be a historically salient approximate value that emerges quite naturally when typical truncations of common distributions are analyzed. As seen in Section \ref{sec:canonicalprinciples}, generalized principles near 0.2/0.8 indeed seem more likely than more extreme cases, but this does not mean one should treat this specific value as a target. 

A plausible cause for the endurance of the original Pareto principle is that the values of generalized principles satisfied by data seem to often come close to the 0.2/0.8-principle. Allowing a non-well-defined tolerance to the satisfied principle, matches close to the canonical $0.2/0.8$-principle are interpreted as perfect. In reality there is a continuous family of generalized principles, and it would be more enlightening to find the distributional form and the true value of $p$ that the data follows, instead of trying to force the canonical principle into any dataset.

\section{Conclusions}
\label{sec:conclusion}

We have formalized and studied the existence of the (generalized) Pareto principle or the $p/(1-p)$-principle. The principle was formalized in Section \ref{sec:formalization} with the decreasing rearrangement of a density function $\ell(t)$ describing the density of gains on the unit interval. This allowed us to define the satisfied generalized principle unambiguously as the maximal (symmetric) discrepancy of gain densities. Along the way, to arrive at a sensible formalization of the principle, necessary limitations on $\ell(t)$ were discussed and imposed.

After this, explicitly constructed and commonly encountered distribution families were studied in Section \ref{sec:examplesexistence}. Constructed examples illustrated the formal framework, the limitations proposed on $\ell(t)$, and allowed for more general arguments about the minimal level of inequality required to satisfy a specific $p/(1-p)$-principle. It was also shown that for density profiles encountered often in real datasets, \textit{a priori} many of the generalized principles could be satisfied by suitable parameter choices. Limitations of polynomial and power-law distributions were explicitly derived.

Finally, commonly encountered generalized principles were quantified, and the social implications of this framework and its results were discussed in Section \ref{sec:discussion}. For power-laws, the principles satisfied by a limited selection of datasets were below the original 0.2/0.8-principle. For the exponential and normal distributions, as a function of dataset size, the values in the ranges of satisfied generalized principles were considerably close to the canonical 0.2/0.8-principle. 

Multiple lines of generalization and further study lay open. First, a simple follow-up would be to study other commonly encountered (truncated) distributions like the log-normal and gamma distributions, and to map out their patterns of satisfying the generalized Pareto principle. This would contribute to the distribution of usual generalized principles encountered in datasets, allowing a sharper estimate of the most common generalized Pareto principle found in data. A more detailed investigation of the social implications also warrants further study.

One could also generalize the framework to include multi-dimensional data and study the generalized principles in such cases. For example, $\L$ could be defined by integrating an $\ell \in L^1([0,1]^n)$ on a (hyper)cube $[0,p]^n$, or perhaps more naturally on a radius-$p$ hypersphere. In this way, a 1-dimensional cumulative gain function would still be obtained, and the generalized principle could be studied similarly.

Time evolution of real distributions could also be modelled to understand how such evolution shows up in the observed generalized principles. This would shed light on the dynamics of the Pareto imbalances, where models like kinetic exchange models or diffusion models could be used as example cases. One could study how quickly a stable generalized principle is converged into, or whether major fluctuations in such balances exist. We relegate such questions to future work.

\section*{Acknowledgements}

The author thanks Mika Pantzar, Niko Jokela, and Miika Sarkkinen for help and valuable discussions during the preparation of this work.

\appendix

\section{On the limitations of \texorpdfstring{$\ell(t)$}{l(t)}}
\label{app:limitationsoff}

Let us discuss further the reasons for imposing various limitations on $\ell$, as well as the possibility of lifting them under given circumstances.

\subsection{\texorpdfstring{$\ell \in L^1([0,1])$}{l in L1([0,1])} and differentiability of \texorpdfstring{$\L$}{L}}

We required $\ell \in L^1([0,1])$. Satisfying the existence property \ref{rmrk:remark1} requires $\L$ to be continuous, so $\ell \in L^1([0,1])$ is sufficient. However, $L^1([0,1])$ also allows for other desirable properties to exist: First, $\ell$ could in principle possess integrable singularities. For example, in the case of learning over time, such singularities could model ``Eureka moments'' or limiting processes of such, resulting in a divergent momentary gain rate. Second, it is natural to allow $\ell(t)$ to be discontinuous. For example, the existence of a step function gain density profile could result from a limiting process of quickly varying gains. As discussed in the section below, removing padding by zeros can naturally result in discontinuous gain densities.

Hence, we only require $\ell \in L^1([0,1])$ and cannot assume $\L$ to be differentiable, rendering it secondary to the discussion; we can unambiguously define $\L(t)$ in terms of $\ell(t)$, but not the other way around. As gain densities are not defined distributionally, cumulative gain profiles are continuous, and a natural class for them is the class of (absolutely) continuous functions, $C^0_{a}([0,1])$. If we instead assume $\ell$ to be continuous, as is the case with the distribution families in Section \ref{sec:examplesexistence}, then $\L$ is at least once differentiable.

\subsection{Padding by zeros}

Another limitation we must impose is on the length of intervals where $\ell(t) = 0$, that is, on the support of $\ell(t)$. A continuous family of generalized principles is trivial to satisfy if one allows such ``padding by zeros''. The possibility of padding by zeros does seem natural, since there can for example be periods of various lengths when no gains are achieved. However, the argument below shows why such periods must be removed from the density. 

Consider that we want to satisfy the generalized principle with a given $p$ and that there exists a $p' > p$ such that
\begin{equation}
    \int_0^{p'} \dd s \ \ell(s) = 1 - p \ .
\end{equation}
Then, scale (``squeeze'') $\ell(t)$ and pad it with zeros for the rest of the interval: 
\begin{equation}
    \Tilde{\ell}(t) = 
    \begin{cases}
        &\ell\qty(\frac{t}{\K}) \ , \ 0 < \K = p/p' < 1 \\
        &0\ , \ \text{otherwise.} 
    \end{cases} 
\end{equation}
Since $\ell(t)$ is normalized, we need to normalize $\Tilde{\ell}$:
\begin{equation}
    \int_0^1 \dd t \ \Tilde{\ell}\qty(t) = \int_0^\K \dd t \ \ell\qty(\frac{t}{\K}) = \K \int_0^1 \dd u \ \ell(u) = \K  \ .
\end{equation}
Normalization does not affect the ratio of densities in different intervals. Therefore, by padding with zeros, we satisfy the originally sought generalized principle, since now
\begin{equation}
    \frac{1}{\K}\int_0^p \dd t \ \Tilde{\ell}(t) =\frac{1}{\K} \int_0^p \dd t \ \ell\qty(\frac{t}{\K}) = \frac{1}{\K} \K \int_0^{p / \K = p'} \dd  u \ \ell(u) = 1 - p \ .
\end{equation}

Thus, we must limit our attention to gain density functions $\ell(t)$ which have essential support for the full interval $t \in [0,1]$ -- that is, the gain densities cannot equal zero on any finite-length intervals. Formally, $\mu\qty(\{\ell(t) = 0\} \mid t \in [0,1]) = 0$.

\subsection{Negative gains}

In principle, one could consider the possibility of negative gains as well. The question seems to be mostly semantic; take as an example the case of learning: is forgetting something you have learned ``negative learning", or is ``forgetting'' a separate process from learning? Conversely, is obtaining such forgotten information ``learning again'', or rather ``recalling'', a complementary process to forgetting?

Whatever the resolution of such questions, we focus only on non-negative gain densities. Allowing negative values entails accepting the (generalized) Pareto principle in a manner which does not resonate with the intuitive perception of imbalances. If negative rates were allowed, we should also extend the range of $\L(t)$ from $[0,1]$ to $\Rb$. With a decreasing rearrangement performed for $\ell$, there would still exist a unique generalized principle satisfied by such a profile.

\section{Trivialized existence of the generalized principle}
\label{app:rearrangementsproof}

Here we prove Theorem \ref{thm:theorem2} stating, that allowing integration over a countable disjoint union of intervals, we satisfy a family of generalized principles between two well-defined maximal values.

\begin{theorem*}
    For a gain density $\ell \in L^1([0,1])$, allowing the generalized Pareto principle to be satisfied over any countable disjoint union of intervals $\bigcup_{j \in \J} I_j, \ \forall \ j \ I_j \subset [0,1]$, is equivalent to fixing an interval of integration of length $\sum_j \mu(I_j)$ and allowing arbitrary rearrangements of $\ell$ into some $\ell^*$, which is equimeasurable with $\ell$. This also results in the rearranged gain function $\L^*$. \\

    \noindent With such rearrangements allowed, the minimal $p$ for satisfying the generalized principle is obtained by integrating $\ell^*$ on an interval $[0,p^*]$ such that $\L^*(p^*) = 1- p^*$. By symmetry, the maximal value of $p$ satisfying the generalized principle is $1 - p^*$. \\

    \noindent Then, every generalized Pareto principle between $p^*$ and $(1-p^*)$ can be satisfied by some rearrangement.
\end{theorem*}

\noindent \textbf{Proof: }Let $([0,1], \mu)$ be the unit interval with the standard Lebesgue measure, and let $\ell$ be the normalized, non-negative gain density, $\ell \in L^1([0,1])$. For a given $p \in (0, 1/2]$ let $\I_p$ denote the collection of measurable sets such that for all $ I, I \subset [0,1]$ and $\mu(I) = p$. Any $I \in \I_p$ can be approximated arbitrarily well by a countable collection of open sets.

First, allow the generalized principle to be satisfied over an arbitrary countable disjoint union of intervals and assume there exists a measurable set $I \subset [0,1]$, $\mu(I) = p$ such that $\int_I \dd \mu \ \ell = 1 - p$. Since the measure is Lebesgue, it is non-atomic, and there exists a measure-preserving bijection $\tau$ with $\tau([0,p]) = I$ \cite{Halmos:1974}. By a change of variables under measure-preserving maps, we may write
\begin{equation}
    \int_0^p \dd \mu \ \qty(\ell \circ \tau^{-1}) = \int_{\tau([0,p])} \dd \mu \ \ell = \int_I \dd \mu \ \ell = 1- p \ ,
\end{equation}
and therefore we can fix the interval of integration using the rearranged function $\ell^* = \ell \circ \tau^{-1}$.

Next, assume that we want to fix the interval of integration (without loss of generality to $[0,p]$) and that we may perform rearrangements of $\ell$, which satisfies for some measurable set $I \in \I_p$
\begin{equation}
    \int_I \dd \mu  \ \ell = 1 - p  \ .    
\end{equation}
Again by the existence of a measure-preserving bijection $\tau : [0,1] \to [0,1]$ with $\mu(\tau^{-1} I) = \mu(I)$, we construct $\tau$ \cite{Kechris:1995} such that $\tau^{-1}(I) = [0,p]$ and find
\begin{equation}
    \int_I \dd \mu \ \ell = \int_{\tau([0,p])} \dd \mu \ \ell = \int_0^p \dd \mu \ \ell \circ \tau^{-1} = 1 - p \ .
\end{equation}

Next, prove the existence and uniqueness of a minimal value $p^*$ with which the generalized principle can be satisfied for a given distribution $\ell$. 

A decreasing rearrangement of $\ell$, $\ell^* : [0,1] \to [0, \infty]$ is the nonincreasing function that is equimeasurable with $\ell$, that is, for all $s \in \Rb, \ \mu(\{\ell > s\}) = \mu(\{\ell^* > s\})$ \cite{Lieb:2001}. By the Hardy-Littlewood inequality, for all $I \in \I_p$, 
\begin{equation}
    \int_I \dd \mu \ \ell  \  \leq \int_0^p \dd s \  \ell^*(s) \ .
\end{equation}
By the first part, we may work directly with the decreasingly rearranged distribution $\ell^*$. 

By \ref{rmrk:remark1}, for any $\ell$, some generalized principle is satisfied. By symmetry, finding a minimal $p^*$ is equivalent to finding a maximal $(1-p^*)$. By the Hardy-Littlewood theorem, a maximal $1-p^*$ satisfying the other half of the principle can be found by integrating on the interval $[0, p^*]$ such that $\int_0^{p^*} \dd s \ \ell^*(s) = 1 - p^*$. Hence, such a maximal $1 - p^*$ exists and by symmetry, the corresponding minimal $p^*$ exists, defined by this condition. Since $\ell^*$ is monotonic, $p^*$ is unique. \\

Finally, we want to prove that not just the limiting generalized principles with $p^*$ and $1-p^*$, but also all generalized principles with $p \in (p^*, 1-p^*)$ can be satisfied by a suitable rearrangement. The proof is again by the intermediate value theorem. By the second part, there exists a unique minimal $p^*$ with which the principle is satisfied, found from the decreasing rearrangement $\ell^*$. The above claim is trivial with $p^* = \frac{1}{2}$, so assume this to not be the case. Hence, we have to show that the generalized principle can be satisfied for all $p$ with $p^* < p < \frac{1}{2}$. 

Assume $p^* < p < \frac{1}{2}$. Since $\ell^*$ is decreasing and positive, 
\begin{equation}
        \int_0^p \dd t \ \ell^*(t) = p_1 \geq 1 - p^* \ , \qquad \int_{1-p}^1 \dd t \ \ell^*(t)= p_2 \leq p^* \ .
\end{equation}
Hence, we find that integrating over an interval of length $p$ will result in a total mass 
\begin{equation}
    p_1 \geq 1-p^* > p^* \geq p_2 \ .
\end{equation}
Since $p > p^*$, we would like to find an interval with mass $1-p$ such that $1-p^* > 1 - p > p^*$. Define the function evaluating the total mass inside an interval of length $p$ with
\begin{equation}
    M(s) = \int_s^{s + p} \dd t \ \ell^*(t) \ ,
\end{equation}
with $s \in [0, 1-p]$. $M$ is continuous and
\begin{equation}
    \begin{split}
        M(0) = p_1\ , \qquad M(1-p) = p_2 \ .
    \end{split}
\end{equation}
By the intermediate value theorem, there exists some $s \in [0, 1-p]$ at which $M$ obtains the value $1-p$, since $p_1 \geq 1 - p \geq p_2.$ Hence, all generalized principles between $p^*$ and $1-p^*$ may be satisfied if integration over countable disjoint unions is allowed. \ $\square$

\bibliography{biblio.bib}

\end{document}